\documentclass[a4paper,10pt,reqno]{amsart}
          \usepackage{amssymb}
	  \usepackage{amsmath}
          \usepackage{amsfonts}
          \usepackage[english]{babel}
          \usepackage[utf8]{inputenc}
\usepackage[margin=2.5cm]{geometry}

 \usepackage[unicode,colorlinks,plainpages=false,hyperindex=true,bookmarksnumbered=true,bookmarksopen=true,pdfpagelabels,
 pdfauthor={Maxime~Fairon and Tamas~Gorbe},
 pdftitle={Superintegrability of Calogero-Moser systems associated with the cyclic quiver},%
pdfsubject={Superintegrability of CM systems},%
pdfkeywords={Calogero-Moser system, superintegrability, cyclic quiver}]{hyperref}
 \hypersetup{urlcolor=cyan,linkcolor=blue,citecolor=red,colorlinks=true} 
\usepackage{tikz}
\usepackage{doi}
\newcommand{\arXiv}[2]{arXiv:\href{http://arxiv.org/abs/#1}{#1 #2}}
\usepackage{xcolor}
  
\makeatletter
\renewcommand*{\p@section}{\,}
\renewcommand*{\p@subsection}{\S\,}
\renewcommand*{\p@subsubsection}{\S\,}
\makeatother

\newtheorem{thm}{Theorem}[section]
\newtheorem{lem}[thm]{Lemma}
\newtheorem{prop}[thm]{Proposition}
\newtheorem{exmp}[thm]{Example}
\newtheorem{rem}[thm]{Remark}

\numberwithin{equation}{section}

\newcommand{\CC}{\ensuremath{\mathbb{C}}}
\newcommand{\N}{\ensuremath{\mathbb{N}}}

\newcommand{\Z}{\ensuremath{\mathbb{Z}}}

\newcommand{\ii}{\ensuremath{\mathfrak{i}}}
\newcommand{\rmm}{\ensuremath{\mathrm{m}}}
 
\newcommand{\tr}{\operatorname{tr}}
\newcommand{\diag}{\operatorname{diag}}

\newcommand{\Hom}{\operatorname{Hom}}
\newcommand{\End}{\operatorname{End}}

\newcommand{\Mat}{\operatorname{Mat}}
\newcommand{\Id}{\operatorname{Id}}
\newcommand{\Gl}{\operatorname{GL}}
\newcommand{\Rep}{\operatorname{Rep}}
\newcommand{\gl}{\ensuremath{\mathfrak{gl}}}

\newcommand{\VV}{\ensuremath{\mathcal{V}}}

\newcommand\dgal[1]{  \left\{\!\!\left\{#1\right\}\!\!\right\} }
\newcommand\br[1]{\{ #1 \}} 

\def\dfat{\mathbf{d}}

\def\nfat{\mathbf{n}}
\def\ntil{\widetilde{\mathbf{n}}}
\def\lambt{\widetilde {\lambda}}
\def\mut{\widetilde{\mu}}

\def\Cspin{\mathcal{C}_{\nfat,\dfat,\lambt}}
\def\CspinOp{\mathcal{C}_{\nfat,\dfat^{\circ},\lambt^{\circ}}}
\def\Cnfat{\mathcal{C}_{\nfat}}
\def\reg{\mathrm{reg}}  
\def\Cnreg{\CC^n_{\reg}}

\begin{document}

\title{Superintegrability of Calogero-Moser systems associated with the cyclic quiver}   

\author{Maxime Fairon}
 \address[Maxime Fairon]{School of Mathematics and Statistics, University of Glasgow, University Place, Glasgow G12 8QQ, UK}
 \email{Maxime.Fairon@glasgow.ac.uk}

\author{Tam\'{a}s G\"{o}rbe}
 \address[Tam\'{a}s G\"{o}rbe]{Bernoulli Institute for Mathematics, Computer Science and Artificial Intelligence, University of Groningen, P.O. Box 407, 9700 AK Groningen, The Netherlands}
 \email{T.Gorbe@rug.nl}

\begin{abstract}
We study complex integrable systems on quiver varieties associated with the cyclic quiver, and prove their superintegrability by explicitly constructing first integrals. We interpret them as rational Calogero-Moser systems endowed with internal degrees of freedom called spins. They encompass the usual systems in type $A_{n-1}$ and $B_n$, as well as generalisations introduced by Chalykh and Silantyev in connection with the multicomponent KP hierarchy. We also prove that superintegrability is preserved when a harmonic oscillator potential is added.
\end{abstract}

\maketitle

 \setcounter{tocdepth}{2}


\section{Introduction}  \label{S:Intro} 

The integrable $n$-particle systems of Toda \cite{To67}, Calogero-Moser \cite{Cal,Mo75}, and Ruijsenaars-Schneider \cite{RS86} have a remarkable tendency to maintain many of their interesting properties when being extended in various ways\footnote{For a brief overview of these integrable systems, we refer to the introductions of our PhD Theses \cite{F19,G17}.}. These properties include superintegrability (if present) and their connections with one another as well as with other objects, may those be soliton equations, orthogonal polynomials or models in statistical physics. The extensions we have in mind include giving the particles internal degrees of freedom (spin models), replacing the underlying type $A$ root system (boundary potentials) or defining the systems on exotic spaces (e.g. quiver varieties). This paper reinforces the above-mentioned phenomenon by proving the superintegrability of (spin) Calogero-Moser type systems attached to cyclic quivers.

Before delving into the particulars of the systems we are to study, let us define what we mean by superintegrability. For our purposes, a superintegrable Hamiltonian system with $N$ degrees of freedom, that is a $2N$-dimensional symplectic manifold $(M,\omega)$ with a smooth function $H\in C^\infty(M)$ of special importance, has $2N-1$ globally defined, independent constants of motion. Such systems are usually referred to as \emph{maximally superintegrable} in the literature \cite{Win04}. We note that maximal superintegrability is a special form of non-commutative (or degenerate) integrability \cite{Ne72,MF78}. The study of superintegrable systems has a long history with such notable examples as the Kepler problem or the $n$-dimensional isotropic harmonic oscillator \cite{Pe90}, but despite its maturity, the field continues to furnish new developments, see e.g. \cite{BMM20,ELW20,FH,Fo20,Ts20}.

The motivation for this work comes from Chalykh and Silantyev's paper \cite{CS} which generalised the KP hierarchy and (spin) Calogero-Moser type systems to cyclic quivers. A natural question to ask is:
\begin{quote}
\emph{Are these new quiver generalisations of (spin) Calogero-Moser systems superintegrable?}
\end{quote}
Our main result is an \emph{affirmative} answer to this question via an explicit construction.

To help place this work into context, let us give a quick (incomplete) review of previous results on the superintegrability of (spin) CM systems. In 1975/76 Adler \cite{A77} showed the superintegrability of the rational Calogero-Moser Hamiltonian with a harmonic potential added (this variant is also known as the Calogero model). In 1983 Wojciechowski \cite{Woj} proved superintegrability of all Hamiltonians of the rational Calogero-Moser system. In 1988 Ruijsenaars \cite{Ru88} published his scattering theory of rational and hyperbolic CM and RS systems (which implies superintegrability). In 1999 Caseiro-Fran\c{c}oise-Sasaki \cite{CasFS} proved superintegrability of rational CM attached to any finite Coxeter group. In 2003 Reshetikhin \cite{R03} established the degenerate integrability of spin CM systems corresponding to co-adjoint orbits of simple Lie algebras \cite{LX02}. Let us also mention the papers \cite{AF09,FG14} where explicitly formulated constants of motion for the rational RS system were found.

To give a sense of the type of integrable systems we consider, they include (as a special case) the rational $B_n$ spin Calogero-Moser model with an external harmonic oscillator potential whose Hamiltonian reads
\begin{equation} \label{Eq:HamIntro}
H=\frac{1}{2}\sum_{i=1}^np_i^2+\sum_{\substack{i,j=1\\(i<j)}}^n f_{ij}f_{ji}\left[\frac{1}{(x_i-x_j)^2}+\frac{1}{(x_i+x_j)^2} \right]+\frac{\gamma_1}{2}\sum_{i=1}^n\frac{1}{x_i^2}+\frac{\omega^2}{2}\sum_{i=1}^n x_i^2
\end{equation}
with particle momenta and positions $(p_i,x_i)$,  spin variables $f_{ij}$ and arbitrary coupling constants $\gamma_1,\omega$.

\noindent
Note that the variables $f_{ij}$ can be seen as ``collective" spins. For a fixed $d>1$, they depend on $2nd$ (constrained) parameters that are interpreted as $n$ sets of $2d$ spin variables, where one such set is attached to each particle.  

The key idea (inspired by the works \cite{AF09,AFG,Cas02}) that lets us construct the constants of motion required for superintegrability can be summarised as follows. Let $M$ be an arbitrary Poisson manifold (either real or complex) with a Poisson bracket $\br{-,-}$. Then we have the following
\begin{thm} \label{Thm:sup}
Fix a function $H$ on $M$, and assume that there exists a family of functions $(g_j)_{j \in \N}$ such that for all $j\in \N$
\begin{equation*}
 \br{H,g_j}\neq 0\,, \quad  \br{H,\br{H,g_j}}=\alpha_j g_j\,, 
\end{equation*}
for some constants $\alpha_j$.   
\begin{enumerate}
 \item[a)] For any $j,k \in \N$ with  $\alpha_j=\alpha_k$, the function 
\begin{equation}
  C_{j,k}^H:=g_j\, \br{H,g_k} - g_k\, \br{H,g_j} \,, 
\end{equation}
is a first integral of $H$. 
\item[b)] For any $j \in \N$, the function 
\begin{equation}
  \tilde C_{j}^H:=\, \br{H,g_j}^2 - \alpha_j g_j^2 \,, 
\end{equation}
is a first integral of $H$.
\end{enumerate}
\end{thm}
\begin{thm} \label{Thm:supB}
Fix a function $H$ on $M$, and assume that there exist two families of functions $(g_j)_{j \in \N}, (\tilde g_j)_{j \in \N}$ such that 
for all $j\in \N$ 
\begin{equation*}
 \br{H,g_j}= \alpha_j g_j\,, \quad \br{H,\tilde g_j}=\tilde \alpha_j \tilde g_j\,,
\end{equation*}
for some constants $\alpha_j, \tilde \alpha_j$. Then, for any $j,k \in \N$ with $\alpha_j=-\tilde \alpha_k$, the function 
\begin{equation}
  D_{j,k}^H:=g_j \tilde g_k \,, 
\end{equation}
is a first integral of $H$.
\end{thm}

The proofs of these results involve a straightforward use of the Leibniz rule and the assumptions. In fact, Theorems \ref{Thm:sup} and \ref{Thm:supB} hold more generally for derivations, so they can be used in the quantum case, too.

\begin{rem}
In this paper, we adopt the convention $\N=\{0,1,\ldots\}$ and work in the complex setting, that is over the field of complex numbers $\mathbb{C}$.
\end{rem}

The structure of the paper is as follows. In Section \ref{S:CM}, we describe the spinless Calogero-Moser spaces and prove superintegrability for spinless rational Calogero-Moser systems attached to cyclic quivers. Section \ref{S:spinCM} contains the spin generalisation of the results of Section \ref{S:CM}. In Section \ref{SS:Harm}, we prove superintegrability for the (spin) rational Calogero Hamiltonian (i.e. CM particles in a harmonic well) associated with classical Lie algebras. Section \ref{S:Dbr} explains the basics of the main computational tool of the paper, double brackets, and it contains the detailed derivations of formulas used in previous sections. Finally, in Section \ref{S:end}, we conclude the paper with an outlook on possible generalisations and future plans.

\subsection*{Acknowledgements.}
We thank L. Feh\'{e}r for bringing relevant references to our attention.
The work of M.F. was partly supported by a Rankin-Sneddon Research Fellowship of the University of Glasgow. 

\bigskip\noindent
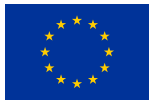
\parbox{.125\textwidth}{\begin{tikzpicture}[scale=.035]
\fill[fill={rgb,255:red,0;green,51;blue,153}] (-27,-18) rectangle (27,18);  
\pgfmathsetmacro\inr{tan(36)/cos(18)}
\foreach \i in {0,1,...,11} {
\begin{scope}[shift={(30*\i:12)}]
\fill[fill={rgb,255:red,255;green,204;blue,0}] (90:2)
\foreach \x in {0,1,...,4} { -- (90+72*\x:2) -- (126+72*\x:\inr) };
\end{scope}}
\end{tikzpicture}} \parbox{.85\textwidth}{This project has received funding from the European Union's Horizon 2020 research and innovation programme under the Marie Sk{\l}odowska-Curie grant agreement No 795471.}

\vspace{0.8cm}


\section{Calogero-Moser system for the cyclic quiver}  \label{S:CM} 

In this section, we consider Calogero-Moser spaces of complex dimension $2n$ associated with cyclic quivers on $m\geq1$ vertices extended by one arrow. Their connection to integrable systems in the simplest case ($m=1$) goes back to Wilson \cite{W}, and has been extended in \cite{CS,GG18}. 

\subsection{Description of the space}  \label{SS:CMdesc}

We omit a detailed introduction to the spaces at hand since they are special cases of the spaces introduced in Section \ref{S:spinCM}. In those notations, we consider $\dfat=(1,0,\ldots,0)$, and put $V:=V_{0,1}$ and $W:=W_{0,1}$.

We fix integers $n,m \geq 1$, and for $I=\Z/m\Z$ we choose a generic  $\lambt=(\lambda_s)\in \CC^I$, see \ref{SS:spinCMdesc} for the precise genericity conditions.  We let $|\lambt|=\sum_{s\in I} \lambda_s$. The Calogero-Moser space $\Cnfat$ is obtained by Hamiltonian reduction from the set of matrices 
\begin{equation*}
 X_s,Y_s\in \Mat_{n \times n}(\CC),\,\,s\in I=\Z/m\Z, \quad V\in \Mat_{1 \times n}(\CC),\, \, W \in \Mat_{n \times 1}(\CC)\,,
\end{equation*}
by requiring the $n$ matrix conditions 
\begin{equation} \label{Eq:Mom}
 X_s Y_s - Y_{s-1} X_{s-1} - \delta_{s,0}\, W V = \lambda_s \Id_{n_s}\,,
\end{equation}
before considering orbits of the action of $\Gl(\nfat)=\prod_{s\in I}  \Gl_{n}(\CC)$ given by  
\begin{equation} \label{Eq:Act}
 g\cdot (X_s,Y_s,W,V)=(g_s X_s g_{s+1}^{-1},g_{s+1}Y_s g_s^{-1},g_0 W,V  g_0^{-1})\,, \quad g=(g_s)\in \Gl(\nfat)\,.
\end{equation}

We consider a first restriction to the subset $\Cnfat^\circ \subset \Cnfat$ where the product $X_0\ldots X_{m-1}\in \Mat_{n \times n}(\CC)$ is diagonalisable, and its diagonal form is given by $\diag(x_1^m,\ldots,x_n^m)$ where  
$(x_1,\ldots,x_n)\in \Cnreg$ for 
\begin{equation} \label{Eq:Cnreg}
 \Cnreg:=\{ (x_1,\ldots,x_n)\in (\CC^\times)^n \mid  x_i^m \neq x_j^m,\, i\neq j\}\,.
\end{equation}
We then choose a representative where $X_s=D$ for each $s\in I$, with $D=\diag(x_1,\ldots,x_n)$. Finally, we look at the subset $\Cnfat' \subset \Cnfat^\circ$ where for such representatives, the vector $W$ has non-zero entries. In $\Cnfat'$, it is an easy exercise to see that we can parametrise any point $(X_s,Y_s,V,W)$ using 
\begin{equation}
 \begin{aligned} \label{Eq:Cn1}
  X_s=\diag(x_1,\ldots,x_n),\,\, s\in I\,, \quad W=(1,\ldots,1)^T,\quad V=-|\lambt|(1,\ldots,1)\,, \\
  Y_s=(Y_s)_{ij}, \quad \text{for } (Y_s)_{ij}=\delta_{ij}p_j+\delta_{ij} \frac{1}{x_i}(\lambda_1+\ldots+\lambda_s) - \delta_{(i\neq j)}\,|\lambt|\frac{x_i^{m-s-1}x_j^s}{x_i^m-x_j^m}\,,
 \end{aligned}
\end{equation}
where $(x_1,\ldots,x_n)\in \Cnreg$ and $(p_1,\ldots,p_n)\in \CC^n$. We can also see that this is unique up to $\Z_m\wr S_n$ action, which acts by permutation of the entries using $S_n$, and by $(x_1,\ldots,x_n)\mapsto(\mu^r x_1,\ldots,\mu^r x_n)$ using $\Z_m$, where $\mu$ is a primitive $m$-th root of unity. The reduced Poisson bracket is canonical and given by 
\begin{equation}
 \br{x_i,x_j}=0\,, \quad \br{x_i,p_j}=\frac1m \delta_{ij}\,, \quad \br{p_i,p_j}=0\,.
\end{equation}

\subsection{Superintegrability} \label{SS:CMint} 
We form the matrix $X\in \Mat_{nm \times nm}(\CC)$ as an $m\times m$ matrix with blocks of size $n\times n$, where the only nonzero blocks are given by placing $X_s$ in position $(s,s+1)$. In the same way, we form $Y\in \Mat_{nm \times nm}(\CC)$ with only nonzero blocks being $Y_s$ placed in position $(s+1,s)$.  (With the notations of \ref{SS:spinCMdesc}, $X=\sum_s X_s$ and $Y=\sum_s Y_s$.) In particular, $X^k$ and $Y^k$ are block diagonal if and only if $k$ is divisible by $m$. 
The functions $\tr Y^{mi}$, $i \in \N$, are trivially Poisson commuting on $\Cnfat$, see Lemma \ref{Lem:dbrYY}. In this section, we are interested in proving that each such function is superintegrable based on the following example.

\begin{exmp}
 In the case $m=1$, we have for $h_i=\frac1i \tr Y^i$, that the functions $(h_i)_{i=1}^n$ define an integrable system such that $h_2$ is the Hamiltonian for the CM system. We note that for any $i \in \N^\times$
 \begin{equation*}
  \br{h_i,\tr X Y^k}=-\tr Y^{k+i-1}\,,
 \end{equation*}
is a first integral of the integrable system. Thus 
\begin{equation} \label{Eq:Woj}
 C_{j,k}^i=\tr(XY^j) \tr(Y^{k+i-1}) - \tr (XY^k) \tr(Y^{j+i-1}) \,,
\end{equation}
is also a first integral of $h_i$ by Theorem \ref{Thm:sup}. This is Wojciechowski's integral  $K^{(i)}_{j+1,k+1}$ \cite{Woj}. 
\end{exmp}

We now fix $m \geq 1$, and set $h^{m,i}=\frac{1}{mi}\tr Y^{mi}$. 
\begin{lem} \label{Lem:WojInt}
 Fix $i \in \N^\times$. For any $j,k\in \N$, the function 
\begin{equation} \label{Eq:Cm}
 C_{j,k}^{m,i}=\tr(XY^{jm+1}) \tr(Y^{(k+i)m}) - \tr (XY^{km+1}) \tr(Y^{(j+i)m}) \,,
\end{equation}
is a first integral of $h^{m,i}$. 
\end{lem}
\begin{proof}
It is proved in  Lemma \ref{Lem:dbrYY} that $\br{h^{m,i},\tr(XY^{jm+1})}= -\tr Y^{m(i+j)}$, which is a first integral. 
So the result follows from Theorem \ref{Thm:sup} a). 
\end{proof}

\begin{prop} \label{Pr:supCM}
 Fix $i \in \N^\times$. Then the function $h^{m,i}$ is maximally superintegrable.
\end{prop}
\begin{proof} 
 It suffices to show that $h^{m,1},\ldots,h^{m,n}$ and $C_{2,1}^{m,i},\ldots,C_{n,1}^{m,i}$ are functionally independent. This can be done as in \cite{Woj}, see the beginning of the proof of Proposition \ref{Pr:supspinCM}. 
\end{proof}

\begin{rem} \label{Rem:CMexp}
The fact that these systems are Liouville integrable appears in \cite[Section V]{CS}, and it is mentioned in \cite[\S 4.4]{GR} for $\lambt=(0,\ldots,0)$. 
Superintegrability in the case  $m=1$ corresponds to the original work of Wojciechowski \cite{Woj}. Indeed, the function  
 \begin{equation}
  h^{1,2}=\frac12 \sum_{i=1}^n p_i^2 - \lambda_0^2 \sum_{\substack{i,j=1\\(i<j)}}^n \frac{1}{(x_i-x_j)^2}\,, 
 \end{equation}
is the usual rational CM Hamiltonian of type $A_{n-1}$. 
 The case $m=2$ is equivalent to the $B_n$ case \cite[Example 5.6]{CS}. Introducing $p_i'=p_i+\frac{\lambda_1}{2x_i}$ so that $(x_i,p_i')$ are canonical coordinates, we can write 
  \begin{equation} \label{Eq:CM-Dn}
  \frac12 h^{2,1}=\frac12\tr(Y_0Y_1)=\frac12 \sum_{i=1}^n (p_i')^2 - \frac{|\lambt|^2}{4} \sum_{\substack{i,j=1\\(i<j)}}^n \left[\frac{1}{(x_i-x_j)^2}+\frac{1}{(x_i+x_j)^2} \right]-\frac{\lambda_1^2}{8}\sum_{i=1}^n \frac{1}{x_i^2}\,, 
 \end{equation}
 which is the rational CM Hamiltonian in type $B_n$, or type $D_n$ if $\lambda_1=0$ \cite{OP}. Superintegrability of rational CM systems associated with arbitrary root systems is established in \cite{CasFS}. 
\end{rem}


\section{Spin Calogero-Moser systems for the cyclic quiver}  \label{S:spinCM}

\subsection{Phase space}  \label{SS:spinCMdesc} 

We now define the general Calogero-Moser spaces associated with cyclic quivers. When there are several framing arrows going either to one vertex of the cyclic quiver, or when the number of framing arrows is the same for all the vertices in the cyclic quiver, these spaces and the corresponding integrable systems were first studied\footnote{Our presentation differs from the original considerations in \cite{CS} as follows :  we take a different convention for the direction of the framing arrows, and we look at representations of the path algebra of the quivers that we consider, not the opposite quivers.} in \cite{CS} and \cite{GG18}. In the case $m=1$, the spaces can be traced back to the works \cite{W2,BP,T15}, where it was established that the systems correspond to the spin CM system due to Gibbons and Hermsen \cite{GH}.

Fix an integer $m \geq 1$ and let $I=\Z_m=\Z/m \Z$. 
When we consider $I$ as a set, we identify it with $\{0,\ldots,m-1\}$ by sending an element $s\in I$ to its representative in $\{0,\ldots,m-1\}$. 
Moreover, fix $\dfat=(d_0,\ldots,d_{m-1})\in \N^{I}$ such that $|\dfat|=\sum_{s\in I} d_s\geq 1$. 
Without loss of generality, we simply assume that $d_0 \geq 1$ while $d_s \in \N$ for $s \in I \setminus \{0\}$.

We consider the cyclic quiver on $m$ arrows with framing corresponding to $\dfat$, which is defined in the following way. Let $Q_\dfat$ be the quiver with vertex set $\widetilde{I}=I \cup \{\infty\}$, and whose edge set consists, for all $s\in I$, of $d_{s}+1$ arrows given by $x_s:s \to s+1$ and $v_{s,\alpha}:\infty \to s$ with $\alpha=1,\ldots,d_s$. (There is no arrow $\infty \to s$ when $d_s=0$.) The double $\bar{Q}_\dfat$ of $Q_\dfat$ then consists of the same vertex set $\widetilde I$, and $2m+2|\dfat|$ arrows given by the ones described above together with $y_s=x_s^\ast:s+1 \to s$, $w_{s,\alpha}=v_{s,\alpha}^\ast: s \to \infty$ for all $1\leq \alpha \leq d_s$ and $s \in I$. 

\begin{rem} \label{remCyc}
We adopt the following conventions for the rest of the text. 
The indices $r,s$ range over $I$. 
When we consider a couple $(s,\alpha)$, for example as index of $v_{s,\alpha}$, we assume that $s\in I$ as we have just explained and $\alpha$ ranges over the set $\{1,\ldots,d_s\}$. We omit such couples when $d_s=0$. 
\end{rem}

\subsubsection{Definition of the space}

We fix $\ntil=(\nfat,1)$ with $\nfat=(n_s)\in \N^I$ such that $|\nfat|=\sum_s n_s>0$. 
A point $\rho\in \Rep(\CC \bar{Q}_\dfat, \ntil)$ consists of the vector space $\VV=(\oplus_{s\in I}  \VV_s) \oplus \VV_\infty$ with $\VV_s=\CC^{n_s}$ for each $s\in I$ and $\VV_\infty=\CC$, together with $2m+2|\dfat|$ matrices given by 
\begin{equation} \label{Eq:MatGen}
 \begin{aligned}
 &X_s \in \Hom(\VV_{s+1},\VV_{s})\,, \quad 
Y_s \in \Hom(\VV_{s},\VV_{s+1})\,, \\ 
&V_{s,\alpha}\in \Hom(\VV_s,\VV_\infty)\,, \quad
W_{s,\alpha}\in \Hom(\VV_\infty, \VV_s)\,,
 \end{aligned} 
\end{equation}
which respectively represent the arrows $x_s,y_s,v_{s,\alpha},w_{s,\alpha}$. We identify the point $\rho$ with the tuple of matrices $(X_s,Y_s,V_{s,\alpha},W_{s,\alpha})$ to ease our discussion. 
We directly see that $\Rep(\CC \bar{Q}_\dfat, \ntil)$ is a smooth affine variety of dimension $2\sum_{s\in I} n_s (n_{s+1}+d_s)$. 

We have a $\Gl(\nfat):=\prod_{s\in I} \Gl_{n_s}(\CC)$ action on $\Rep(\CC \bar{Q}_\dfat,\ntil)$ given by 
\begin{equation} \label{Eq:spinAct}
 g\cdot (X_s,Y_s,W_{s,\alpha},V_{s,\alpha})=(g_s X_s g_{s+1}^{-1},g_{s+1}Y_s g_s^{-1},g_s W_{s,\alpha},V_{s,\alpha} g_s^{-1})\,, \quad g=(g_s)\in \Gl(\nfat)\,.
\end{equation}
Following e.g. Van den Bergh \cite{VdB1}, the complex manifold $\Rep(\CC \bar{Q}_\dfat, \ntil)$ admits a Poisson bracket $\br{-,-}$ given by 
\begin{equation} \label{Eq:PB}
 \br{(X_r)_{ij},(Y_s)_{kl}}= \delta_{rs} \delta_{kj}\delta_{il}\,, \quad 
 \br{(V_{r,\alpha})_{j},(W_{s,\beta})_{k}}= \delta_{rs} \delta_{\alpha\beta} \delta_{kj}\,,
\end{equation}
and which is zero on any other pair of entries of the matrices \eqref{Eq:MatGen}. Moreover, it is endowed with a moment map $\mut$ with value in 
$\gl(\nfat):=\prod_{s\in I}  \gl_{n_s}(\CC)$ given by 
\begin{equation} \label{Eq:spinMom}
 \mut=\sum_{s\in I}  \mu_s\,, \quad \mu_s= X_s Y_s - Y_{s-1} X_{s-1} - \sum_{1 \leq \alpha \leq d_s} W_{s,\alpha} V_{s,\alpha}\in \End(\VV_s)\,,
\end{equation}
where we omit the final sum in $\mu_s$ if $d_s=0$.  

Fix $\lambt=(\lambda_s)\in \CC^I$ and denote by $\lambt \cdot \Id \in \gl(\nfat)$ the element with blocks  $\lambda_s \Id_{n_s} \in \gl_{n_s}(\CC)$. 
Then, the slice $\mut^{-1}(\lambt \cdot \Id)$ corresponds to imposing the $m$ equations 
\begin{equation} \label{Eq:spinMom2}
 X_s Y_s - Y_{s-1} X_{s-1} - \sum_{1 \leq \alpha \leq d_s} W_{s,\alpha} V_{s,\alpha} = \lambda_s \Id_{n_s}\,,
\end{equation}
from which it follows by taking traces that $\sum_{s\in I}\sum_{1 \leq \alpha \leq d_s}  V_{s,\alpha} W_{s,\alpha}=-\sum_s \lambda_s n_s =: -\lambt \cdot \nfat$. 
Using Hamiltonian reduction, it follows that the GIT quotient $\Cspin=\mut^{-1}(\lambt \cdot \Id)/\!/\Gl(\nfat)$ is a Poisson variety. 
The space hence obtained is a quiver variety : it is the GIT quotient for the $\Gl(\nfat)$ action \eqref{Eq:spinAct} on the representation space associated with a deformed preprojective algebra of $Q$ with parameter $(\lambt,-\lambt \cdot \nfat)$. 

From now on, we further assume that $\nfat=(n,\ldots,n)$ for some $n \in \N^\times$, and we simply denote $\Cspin$ by $\Cnfat$. Then, 
$\Cnfat$ is a non-empty smooth variety which coincides with the set-theoretic orbit space $\mut^{-1}(\lambt \cdot \Id)/\Gl(\nfat)$ provided that the regularity conditions 
\begin{equation} \label{Eq:spinReg}
\lambda_0+\ldots + \lambda_{m-1}\neq0\,, \quad \text{ and }\quad 
k (\lambda_0+\ldots +\lambda_{m-1}) \neq \lambda_r+\ldots+ \lambda_{s-1},\,\,\, k \in \Z,\, 1\leq r<s\leq m-1\,,
\end{equation}
are satisfied, see \cite[Proposition 3]{BCE} or \cite[Theorem 1.2]{CB01}. Note that $\Cnfat$ has dimension $2n |\dfat|$. 

\subsubsection{Local description} \label{sss:Spinloc}

We consider the open subspace $\Cnfat^\circ\subset \Cnfat$ where the product $X_{0}\ldots X_{m-1}$ is invertible with distinct eigenvalues $x_1^m,\ldots,x_n^m$. We pick $m$-th roots $(x_i)$ of the eigenvalues, and by construction of $\Cnfat^\circ$ these elements take value in $\Cnreg$ \eqref{Eq:Cnreg}. We can use the $\Gl(\nfat)$ action to pick any representative such that 
$X_s=\diag(x_1,\ldots,x_n)$ for each $s\in I$, and there remains an overall action by the normaliser $\mathcal N$ of the diagonal subgroup $(\CC^\times)^n\subset\Gl_n(\CC)$ seen as a subgroup of $\Gl(\nfat)$ through $\mathcal N\ni h \mapsto \prod_{s\in I}h \in \Gl(\nfat)$. 

We then define the open subspace $\Cnfat' \subset \Cnfat^\circ$ where for one (hence any) such representative, the vector $\sum_{1\leq \alpha \leq d_0}W_{0,\alpha}$ has non-zero entries. We can then act by a diagonal matrix to find a representative such that $\sum_{1\leq \alpha \leq d_0}W_{0,\alpha}=(1,\ldots,1)^\top$.  This representative is unique up to a $\Z_m\wr S_n$ action described below. Note that $\Cnfat'$ contains the subspace defined in \ref{SS:CMdesc}. 

In this way, we can characterise a point of $\Cnfat'$ by the $2n+2n|\dfat|$ variables $(x_i,p_i,v_{s,\alpha,i},w_{s,\alpha,i})$ such that  $(x_i)\in \Cnreg$, together with the $2n$ constraints 
\begin{equation} \label{Eq:constr}
 \sum_{1 \leq \alpha \leq d_0} w_{0,\alpha,i}=1\,, \qquad 
 \sum_{s \in I} f_{ii}^{(s)}=-|\lambt|\,, \,\, \text{ where }
  f_{ij}^{(s)}:=\sum_{1\leq \alpha \leq d_s} w_{s,\alpha,i}v_{s,\alpha,j}\,,
\end{equation}
by considering the following matrices  
\begin{equation} \label{Eq:spinCoord}
\begin{aligned}
  &X_s=\diag(x_1,\ldots,x_n)\,, \quad 
  (W_{s,\alpha})_i=w_{s,\alpha,i}\,, \quad (V_{s,\alpha})_i=v_{s,\alpha,i}\,, \\
  &(Y_s)_{ii}= p_i +  \frac{1}{x_i} \sum_{0\leq r \leq s} (\lambda_r+f_{ii}^{(r)}) + \frac{1}{x_i} \sum_{r\in I}\frac{r-m}{m} (\lambda_r+f_{ii}^{(r)})\,, \\
  &(Y_s)_{ij}= \sum_{0\leq r \leq s}\frac{x_i^{m-1+r-s}x_j^{s-r}}{x_i^m-x_j^m} f_{ij}^{(r)}
  + \sum_{s<r \leq m-1}\frac{x_i^{r-s-1}x_j^{m+s-r}}{x_i^m-x_j^m} f_{ij}^{(r)}\,, \quad \text{for }i\neq j\,.
\end{aligned}
\end{equation}
This choice is unique up to $\Z_m\wr S_n$ action, which acts by permutation of the entries using $S_n$, and by $(x_1,\ldots,x_n)\mapsto(\mu^r x_1,\ldots,\mu^r x_n)$ using $\Z_m$ where $\mu$ is a primitive $m$-th root of unity. It is easy to see that we have the normalisation 
\begin{equation} \label{EqYsii}
 \sum_{s\in I}(Y_s)_{ii}=m p_i\,.
\end{equation}
In the case $\dfat=(1,0,\ldots,0)$, we can recover \eqref{Eq:Cn1} from \eqref{Eq:spinCoord} by shifting each variable $p_i$ by a multiple of $x_i^{-1}$, since we have $f_{ij}^{(0)}=-|\lambt|$ while $f_{ij}^{(s)}=0$ for $s\neq 0$. 
In the case $\dfat=(d,\ldots,d)$,  our choice of parametrisation is similar to \cite[(6.24-6.25)]{CS}, with the addition of the first $n$  constraints in \eqref{Eq:constr} due to our choice of a finite residual gauge fixing.

\begin{lem}
 The Poisson bracket evaluated on the  $2n+2n|\dfat|$ variables $(x_i,p_i,v_{s,\alpha,i},w_{s,\alpha,i})$ is given by 
 \begin{subequations}
  \begin{align}
   &\br{x_i,x_j}=0\,, \quad \br{x_i,p_j}=\frac1m \delta_{ij}\,, \quad \br{p_i,p_j}=0\,,  \label{Eq:brCo1}\\
   &\br{x_i,v_{s,\beta,j}}=0\,,\,\,\br{x_i,w_{s,\beta,j}}=0\,, \quad \br{p_i,v_{s,\beta,j}}=0\,,\,\,\br{p_i,w_{s,\beta,j}}=0 \,, \label{Eq:brCo2}\\
   &\br{v_{r,\alpha,i},v_{s,\beta,j}}=\delta_{ij} (\delta_{0r}v_{s,\beta,j}-\delta_{0s}v_{r,\alpha,i})\,, \label{Eq:brCo3} \\
   &\br{v_{r,\alpha,i},w_{s,\beta,j}}=\delta_{rs}\delta_{\alpha,\beta}\delta_{ij}-\delta_{0r}\delta_{ij}w_{s,\beta,j}\,, \quad 
   \br{w_{r,\alpha,i},w_{s,\beta,j} }=0 \,. \label{Eq:brCo4}
  \end{align}
 \end{subequations}
\end{lem}
\begin{proof}
This result is a direct application of Lemma \ref{Lem:dbrCoord}. To see this, we note that the following expressions can be written in terms of the local variables on $\Cnfat'$ 
\begin{subequations}
 \begin{align}
\tr X^{km}=& m\sum_{j=1}^n x_j^{km}\,, \quad \tr YX^{km+1}= m\sum_{j=1}^n p_jx_j^{km+1}\,,\label{Eq:CoordBr1} \\
\hat{t}^k_{r\alpha,s\beta}=&\tr W_{r,\alpha}V_{s,\beta}X^{km+r-s}
=\sum_{j=1}^n  w_{r,\alpha,j}v_{s,\beta,j}x_j^{km+r-s}\,. \label{Eq:CoordBr2} 
 \end{align}
\end{subequations}
In particular, we have that 
\begin{equation}
 \sum_{\alpha=1}^{d_0}\hat{t}^k_{0\alpha,s\beta}=\sum_{j=1}^n  v_{s,\beta,j}x_j^{km+r-s}\,.
\end{equation}
It is then a standard computation to see that \eqref{Eq:dbrCa}--\eqref{Eq:dbrCb} written in coordinates yield \eqref{Eq:brCo1}. After these identities are established, we also get  from \eqref{Eq:dbrCc}--\eqref{Eq:dbrCd} that \eqref{Eq:brCo2} holds.

Next, using \eqref{Eq:dbrCe} with $r=r'=0$ and summing over all $\alpha,\alpha'\in \{1,\ldots,d_0\}$, we find the identity \eqref{Eq:brCo3}.  Taking $r'=0$ and summing over $\alpha'$ also in \eqref{Eq:dbrCe}, we find the first identity in  \eqref{Eq:brCo4}. Finally, we can use these Poisson brackets and \eqref{Eq:dbrCe} for arbitrary $r,r'$ to obtain the second equality in \eqref{Eq:brCo4}. 
\end{proof}
\begin{rem}
 The complicated Poisson brackets appearing in \eqref{Eq:brCo3}--\eqref{Eq:brCo4} are due to the gauge fixing. Indeed, take the $2n+2n|\dfat|$ complex Darboux coordinates 
 \begin{equation*}
  x_i,p_i,\bar{v}_{s,\alpha,i},\bar{w}_{s,\alpha,i}\,,
 \end{equation*}
with non-zero Poisson bracket given by 
\begin{equation}
 \br{x_i,p_j}=\frac1m \delta_{ij}\,, \quad \br{\bar{v}_{r,\alpha,i},\bar{w}_{s,\beta,j}}=\delta_{rs}\delta_{\alpha\beta}\delta_{ij}\,.
\end{equation}
If we restrict our attention to the variables $(x_i,p_i)$ and 
\begin{equation}
w_{s,\alpha,i}=\bar{w}_{s,\alpha,i} D_i^{-1},\,\,
 v_{s,\alpha,i}= \bar{v}_{s,\alpha,i} D_i,\quad D_i:=\sum_{\alpha=1}^{d_0} \bar{w}_{0,\alpha,i}\,,
\end{equation}
we note that $\sum_{1\leq \alpha \leq d_0} w_{0,\alpha,i}=1$, while the Poisson bracket takes the form \eqref{Eq:brCo1}--\eqref{Eq:brCo4}. Furthermore, the elements 
\begin{equation} \label{Eq:fCas}
 \sum_{s\in I} f_{jj}^{(s)}:=\sum_{s\in I} \sum_{1\leq \alpha \leq d_s} w_{s,\alpha,j}v_{s,\alpha,j}\,,
\end{equation}
are Casimirs. Fixing the values of the functions in \eqref{Eq:fCas} to $-|\lambt|$, we get the the variables introduced on $\Cnfat'$ with the constraints \eqref{Eq:constr}. 
\end{rem}

\subsection{Superintegrability} \label{SS:spinCMint}

Let $X:=\sum_s X_s$ and $Y:= \sum_s Y_s$. We first recall the following trivial result. 
\begin{lem}\label{Lem:YCM}
The functions $\tr(Y^{km})$ are Poisson commuting. 
\end{lem}
The next result follows from Lemma \ref{Lem:dbrYY}. 
\begin{lem} \label{L:supY}
 Fix $s,r \in I$ such that $d_s,d_r \neq 0$. Let $\rho_{r,s}$ be the representative of $s-r$ in $\{0,\ldots,m-1\}$. 
 Then, for any $1\leq \alpha \leq d_r$, $1\leq \beta \leq d_s$ and $k \in \N$, the function 
 \begin{equation}
  t_{r \alpha, s\beta}^{k}\,=\,\tr \left(W_{r,\alpha}V_{s,\beta}  Y^{km+\rho_{r,s}}  \right)
 \end{equation}
 Poisson commute with $h^{m,i}=\frac{1}{mi}\tr Y^{mi}$. 
\end{lem}
The generalisations of Wojciechowski's first integrals $C_{j,k}^{m,i}$ defined in \eqref{Eq:Cm} remain first integrals of $h^{m,i}$ if we add framing arrows. Indeed, it suffices to reproduce the proof of Lemma \ref{Lem:WojInt} in that case. 

 \begin{prop} \label{Pr:supspinCM}
 Fix $i \in \N^\times$. The function $h^{m,i}$ is maximally superintegrable.
\end{prop}
\begin{proof} 
Let us use the full notation $\Cspin$ of the space to emphasise the dependence on the framing $\dfat$ and the parameter $\lambt$. We form $\dfat^\circ=(d_{m-1},\ldots,d_0)$ and $\lambt^\circ=-\lambt$, noting that $\lambt^\circ$ satisfies \eqref{Eq:spinReg} just as $\tilde{\lambda}$ does.  We can then define the space $\CspinOp$ associated with $\dfat^\circ,\lambt^\circ$, which admits a local description on a dense subspace $\CspinOp'$ by \ref{sss:Spinloc}. 
We can take the $2n|\dfat|=2n|\dfat^\circ|$ elements 
\begin{equation} \label{Eq:coordOp}
 x_j,\,\,p_j,\,\, v_{s,\alpha,j},\,\, w_{s,\alpha,j},\quad j=1,\ldots,n,\,\, (s,\alpha)\neq (0,1)\,,
\end{equation}
as coordinates on $\CspinOp'$. Indeed, in view of the constraints \eqref{Eq:constr}, we can see the $(w_{0,1,j})$ as functions of the variables in \eqref{Eq:coordOp}, and the same holds for the $(v_{0,1,j})$ generically. 

As in \cite[Proposition 6.7]{CS}, we note that there exists a diffeomorphism\footnote{This is not a Poisson isomorphism.} $\Psi:\CspinOp \to \Cspin$ given by 
\begin{equation}
 \Psi(X_s)=Y_{m-s-1},\,\, \Psi(Y_s)=X_{m-s-1},\,\, \Psi(V_{s,\alpha})=V_{m-s,\alpha},\,\,\Psi(W_{s,\alpha})=-W_{m-s,\alpha}\,.
\end{equation}
 We can write the following local expressions on $\CspinOp'$ using \eqref{Eq:spinCoord} and \eqref{EqYsii}
\begin{subequations}
 \begin{align}
\Psi^\ast(\tr Y^{km})=&m \tr(X_0\ldots X_{m-1})^k = m\sum_{j=1}^n x_j^{km}\,,  \label{Eq:locXY1}\\
\Psi^\ast(\tr XY^{km+1})=&\sum_{s=0}^{m-1} \tr(Y_s X^{km+1})  =m\sum_{j=1}^n p_j x_j^{km+1}  \,, \label{Eq:locXY2} \\
\Psi^\ast(t_{r\alpha,s\beta}^k)=&-\tr(W_{m-r,\alpha}V_{m-s,\beta}X^{k+m\rho_{r,s}})
=-\sum_{j=1}^n w_{m-r,\alpha,j}v_{m-s,\beta,j}x_j^{mk+\rho_{r,s}}\,, \label{Eq:locXY3}
 \end{align}
\end{subequations}
with the notations of Lemma \ref{L:supY}. (Here, we see the different matrices as endomorphisms of the vector space $\VV=(\oplus_{s\in I}  \VV_s) \oplus \VV_\infty$.) 

It is clear that the functions \eqref{Eq:locXY1} with $k=1,\ldots,n$ are functionally independent, since their Jacobian matrix with respect to the coordinates $(x_1,\ldots,x_n)$ is invertible as $(x_j)\in \Cnreg$. 
(Without loss of generality, we can replace one of these functions by the Hamiltonian of interest $h^{m,i}$.)
We also note that the functions \eqref{Eq:locXY2} with $k=1,\ldots,n$ can be used as coordinates instead of $(p_1,\ldots,p_n)$ since the Jacobian matrix with entries 
\begin{equation}
 \frac{\partial \Psi^\ast(\tr XY^{km+1})}{\partial p_j}=m\,x_j^{km+1}\,,
\end{equation}
is invertible on $\CspinOp'$. It then follows that the functions $C_{k,1}^{m,i}$ \eqref{Eq:Cm} with $k=2,\ldots,n$ provide another $n-1$ functionally independent first integrals of $h^{m,i}$ due to the identity 
\begin{equation}
 \frac{\partial \Psi^\ast C_{k,1}^{m,i}}{\partial \Psi^\ast(\tr XY^{jm+1})}=\delta_{kj} \Psi^\ast\tr(Y^{(i+1)m})-\delta_{1,j} \Psi^\ast\tr(Y^{(i+1)m}) \,,
\end{equation}
and the fact that $\Psi^\ast\tr(Y^{(i+1)m})$ is generically nonzero on $\CspinOp'$. Thus, we have $2n-1$ first integrals of $h^{m,i}$ whose Jacobian matrix taken with respect to the coordinates $(q_j,p_j)$ is invertible.  We need another $2n(|\dfat|-1)$ first integrals in order to get the desired $\dim(\Cspin)-1$ functionally independent first integrals of $h^{m,i}$.

Assume that $|\dfat|>1$ from now on, otherwise the proof can be concluded here. 
We will find $2n(|\dfat|-1)$ first integrals depending on the coordinates \eqref{Eq:coordOp} with the exception of the $(p_j)$, such that their Jacobian matrix taken with respect to the last $2n(|\dfat|-1)$ coordinates in \eqref{Eq:coordOp} is invertible. The functional independence of these new functions and the previous $2n-1$ ones will then follow from this result. 

We note that the $n$ first integrals  
\begin{equation}
\sum_{\alpha=1}^{d_0} \Psi^\ast(t_{0\alpha,s\beta}^k)=-\sum_{j=1}^n v_{m-s,\beta,j}x_j^{m(k+1)-s}\,,\quad k=1,\ldots,n\,,
\end{equation}
with $(s,\beta)\neq(0,1)$ only depend on the $2n$ coordinates $(q_j,v_{m-s,\beta,j})$. It is straightforward to check that their Jacobian matrix taken with respect to $(v_{m-s,\beta,j})$ is invertible, so that we get a total of $n(|\dfat|-1)$ additional first integrals which are all functionally independent.

There exists $s_+\in \{0,\ldots,m-1\}$ such that $s_+$ is the maximal index for which $d_s\neq 0$. Since $|\dfat|>1$, the pair $(s_+,d_{s_+})\neq (0,1)$ is such that $(v_{m-s_+,d_{s_+},j})$ are $n$ coordinates on  $\CspinOp'$ from the set \eqref{Eq:coordOp} by construction. Next, we note that the $n$ first integrals 
\begin{equation}
 \Psi^\ast(t_{r\alpha,s_+ d_{s_+}}^k)=-\sum_{j=1}^n w_{m-r,\alpha,j}\,v_{m-s_+,d_{s_+},j}\,x_j^{mk+\rho_{r,s_+}}\,,\quad k=1,\ldots,n\,, 
\end{equation}
with $(r,\alpha)\neq(0,1)$ only depend on the $3n$ coordinates $(q_j,w_{m-r,\alpha,j},v_{m-s_+,d_{s_+},j})$.  Since the last $n$ coordinates can be taken to be nonzero at a generic point, we get that  the Jacobian matrix 
\begin{equation}
 \frac{\partial \Psi^\ast(t_{r\alpha,s_+ d_{s_+}}^k)}{\partial w_{m-r,\alpha,j}}=- v_{m-s_+,d_{s_+},j}\,x_j^{mk+\rho_{r,s_+}}\,,
\end{equation}
is invertible, providing another  $n(|\dfat|-1)$ functionally independent first integrals. 
\end{proof}

\begin{rem}
 We have in fact an explicit integration for the flow of $h^{m,i}$ on the unreduced space  $\Rep(\CC \bar{Q}_\dfat, \ntil)$. It follows easily from the following form of the Hamiltonian vector field 
 \begin{equation*}
  \dot Y=0,\,\, \dot W_{s,\alpha}=0,\,\, \dot V_{s,\alpha}=0,\,\, \dot X=Y^{im-1}. 
 \end{equation*}
 This is computed using \eqref{Eq:PB}.
\end{rem}

\begin{rem}
 We can easily verify that the functions $(h^{m,i},t_{s \alpha,s \alpha}^i)$ with $1 \leq i \leq n$ and all possible $(s,\alpha)$ are pairwise Poisson commuting. One can further show that we can form a Liouville integrable system, e.g. by removing  the $(t_{01,01}^i)$ from these functions and then prove the functional independence of the remaining elements as in Proposition \ref{Pr:supspinCM}. This choice of functions is different from the one considered in \cite{CS} which is related to the KP hierarchy. 
\end{rem}


\section{Harmonic CM system} \label{SS:Harm}

In this section, we fix $\omega\in \CC^\times$ and we consider the Hamiltonian $H_\omega=\frac12 \tr (Y^2+\omega^2 X^2)$. We note that it can only be nonzero if $m=1$ or $m=2$. In those cases, we can remark the following result, see Lemma \ref{Lem:dbrLL} for its proof. 
\begin{lem}\label{Lem:Harm}
Let $m=1$ or $m=2$, and set $L=Y^2+\omega^2 X^2$. Then the functions $\tr(L^k)$ are Poisson commuting. 
\end{lem}

\subsection{Non-spin case}

We work over the space $\Cnfat$ as in Section \ref{S:CM}. We note that on the subset $\Cnfat'$ we can write 
\begin{equation*}
H_\omega=\left\{
 \begin{array}{cc}
  h^{1,2}+\frac{\omega^2}{2}\sum_i x_i^2& m=1\,, \\
  h^{2,1}+\omega^2\sum_i x_i^2& m=2\,, \\
 \end{array}
 \right. 
\end{equation*}
so that we can obtain the CM Hamiltonians with harmonic term of type $A_{n-1},B_n$ and $D_n$ by Remark \ref{Rem:CMexp}. For the root system $A_{n-1}$, it was originally introduced by Calogero in the quantum case \cite{Cal}.

\begin{lem}\label{Lem:HarmInt1}
Let $m=1$. 
For any $k,j \in \N$, the function 
\begin{equation}\label{Eq:Harm}
  C_{k,j}^{(\omega,1)}=\tr (XL^k) \tr(YL^j) - \tr(XL^j) \tr(YL^k)\,. 
\end{equation}
is a first integral of $H_\omega$. 
\end{lem}
\begin{proof}
 For $g_{1,j}=\tr X L^j$, we note that $ \br{\tr L,\br{\tr L,g_{1,j} }} = -4\omega^2 g_{1,j}$ by Lemma \ref{Lem:dbrLLm1}. 
 Hence it suffices to apply  Theorem \ref{Thm:sup} a).
\end{proof}

\begin{lem}\label{Lem:HarmInt2}
Let $m=2$. 
For any $k,j \in \N$, the function 
\begin{equation}\label{Eq:Harmbis}
  C_{k,j}^{(\omega,2)}=g_k \,\br{\tr L,g_j} - g_j \,\br{\tr L,g_k} \,, \quad g_k:=\tr\big((XY+YX)L^k\big)\,,
\end{equation}
is a first integral of $H_\omega$. 
\end{lem}
\begin{proof}
  We apply Theorem \ref{Thm:sup} a) to $g_j=\tr(XY L^j+YXL^j)$ since  
$ \br{\tr L,\br{\tr L,g_j}} = -16\omega^2 g_j$  
by Lemma \ref{Lem:dbrLLm1}. 
\end{proof}

\begin{prop} \label{Pr:HarmSup1}
Let $m=1$ or $m=2$. 
 The harmonic CM Hamiltonian $H_\omega$ is maximally superintegrable for generic values of $\omega$. 
\end{prop}
\begin{proof}
We first assume that $m=1$. We note that  as $\omega \to 0$, 
\begin{equation}
 C_{k,j}^{(\omega,1)}\to \tr (XY^{2k})\tr(Y^{2j+1})-\tr(XY^{2j})\tr(Y^{2k+1})\,,
\end{equation}
and the latter is just Wojciechowski's function $C^{2}_{2k,2j}$ in \eqref{Eq:Woj}. Therefore,  the functions 
\begin{equation} \label{Eq:HInd1}
  \tr L, \ldots, \tr L^n, C_{2,1}^{(\omega,1)},\ldots,C_{n,1}^{(\omega,1)}\,, 
\end{equation}
degenerate in the limit $\omega \to 0$ to the functions 
\begin{equation*}
  \tr Y^2, \ldots, \tr Y^{2n}, C^{2}_{4,2},\ldots, C^{2}_{2n,2}\,,
\end{equation*}
which can be shown to be independent as in the proof of Proposition \ref{Pr:supCM}. Thus the functions in \eqref{Eq:HInd1} are independent for generic values of $\omega$.

Next, assume that $m=2$. We note that  as $\omega \to 0$, 
\begin{equation}
 \frac{-1}{8}C_{k,j}^{(\omega,2)}\to \tr (XY^{2k+1})\tr(Y^{2(j+1)})- \tr(XY^{2j+1})\tr(Y^{2(k+1)})\,,
\end{equation}
and the latter is the function $C^{2,1}_{k,j}$ in \eqref{Eq:Cm}. 
Therefore,  the functions 
\begin{equation}
  \tr L, \ldots, \tr L^n, C_{2,1}^{(\omega,2)},\ldots,C_{n,1}^{(\omega,2)}\,, 
\end{equation}
degenerate in the limit $\omega \to 0$ to independent functions as in the previous case, so we can conclude.
\end{proof}

\begin{rem}
In the real setting, additional first integrals that yield the superintegrability of the harmonic CM system in type $A_{n-1}$ have been obtained by Adler  \cite[Section 4]{A77}. They are given as the real part of some complex-valued functions, so that we could not directly use them in our setting. 
\end{rem}

\subsection{Spin case}

We work over the space $\Cspin$ where $\nfat=(n,\ldots,n)$ for some $n \in \N^\times$, as in  Section \ref{S:spinCM}. We set $L_{\pm}:=Y \pm \ii \omega X$ with $\ii=\sqrt{-1}$.

\begin{lem}
Fix $m=1$ or $m=2$, and let $s,r \in I$ be such that $d_s,d_r \neq 0$. Let $\rho_{r,s}$ be the representative of $s-r$ in $\{0,\ldots,m-1\}$. Then, for any $1\leq \alpha \leq d_r$, $1\leq \beta \leq d_s$ and $k \in \N$, the function
 \begin{equation*}
  t^{(\omega,k)}_{r\alpha ,s\beta} = \tr [W_{r,\alpha}V_{s,\beta} L_+^{mk+\rho_{r,s}}]\, \tr [W_{r,\alpha}V_{s,\beta} L_-^{mk+\rho_{r,s}}]\,,
 \end{equation*}
 Poisson commutes with $H_{\omega}$. 
\end{lem}
\begin{proof} For fixed $r,s,\alpha,\beta,k$, we denote  
\begin{equation*}
g= \tr [W_{r,\alpha}V_{s,\beta} L_+^{mk+\rho_{r,s}}]\,, \quad 
\tilde g = \tr [W_{r,\alpha}V_{s,\beta} L_-^{mk+\rho_{r,s}}]\,.
\end{equation*}
Then by Lemma \ref{Lem:dbrLL}
\begin{equation*}
 \br{\tr L,g}=-2\ii \omega (mk+\rho_{r,s}) g\,, \quad \br{\tr L,\tilde g}=+2\ii \omega (mk+\rho_{r,s}) \tilde g\,,
\end{equation*}
so the desired statement directly follows from Theorem \ref{Thm:supB}. 
\end{proof}

\begin{prop} \label{Pr:HarmSup2}
Let $m=1$ or $m=2$.  The spin harmonic CM Hamiltonian $H_\omega$ is maximally superintegrable for generic values of $\omega$. 
\end{prop}
\begin{proof}
We already obtained $2n-1$ functionally independent elements as part of Proposition \ref{Pr:HarmSup1}. 
Next, we note that for $\omega \to 0$, we have $t^{(\omega,k)}_{r\alpha ,s\beta}\to (t^k_{r \alpha,s\beta})^2$. So we can  use the functions $(t^{(\omega,k)}_{r\alpha ,s\beta})$ to construct an additional $2n(|\dfat|-1)$ functions such that, by adapting the proof of Proposition \ref{Pr:supspinCM}, we get $2n|\dfat|-1$ elements which degenerate to functionally independent elements  as $\omega \to 0$. We can then conclude.  
\end{proof}

\begin{exmp}
When $m=2$, the Hamiltonian of interest is $H_\omega=\tr(Y_0Y_1)+\omega^2 \tr(X_0 X_1)$. 
In the coordinates described in \ref{sss:Spinloc}, we can write  
  \begin{equation} 
  \begin{aligned} 
\frac12 H_\omega
=&\frac12 \sum_{i=1}^n \left( p_i^2 - \frac{(\lambda_1 + f_{ii}^{(1)})^2}{4 x_i^2} \right) 
-\frac12 \sum_{\substack{i,j=1\\(i\neq j)}}^n \frac{x_i x_j}{(x_i^2-x_j^2)^2} (f_{ij}^{(0)}f_{ji}^{(1)}+f_{ij}^{(1)}f_{ji}^{(0)}) \\
&-\frac12 \sum_{\substack{i,j=1\\(i\neq j)}}^n \frac{x_i^2}{(x_i^2-x_j^2)^2} 
( f_{ij}^{(0)}f_{ji}^{(0)}+ f_{ij}^{(1)}f_{ji}^{(1)} )
+ \frac{ \omega^2}{2} \sum_{i=1}^n x_i^2 \,.
  \end{aligned}
 \end{equation}
In the case $\dfat=(d_0,0)$, we get that $f_{ij}^{(1)}=0$ for all indices and we obtain \eqref{Eq:HamIntro}  upon setting $f_{ij}=\sqrt{-1}f_{ij}^{(0)}/2$ and $\gamma_1=-\lambda_1^2/4$.   
Furthermore, in the case $\dfat=(1,0)$ the constraints \eqref{Eq:constr} yield that $f_{ij}^{(0)}=-|\lambt|$ for all indices, and we recover \eqref{Eq:CM-Dn} when $\omega=0$. 
\end{exmp}

 
\section{Double brackets and computations}  \label{S:Dbr} 

\subsection{Motivating double brackets} 

For researchers in the field of integrable systems, double brackets can be introduced as an analogous approach to finding a Lax matrix and an $r$-matrix with a different type of derivation rules. To understand this analogy, let us recall that the $r$-matrix approach can be simplified as finding a matrix $L\in \gl_n(\CC)$ and an element $r\in \gl_n(\CC)\otimes \gl_n(\CC)$ such that for a given Poisson bracket, we can write 
\begin{equation} \label{Eq:rL}
 \br{L\overset{\otimes}{,}L}=[r,L_1]-[r^\circ,L_2]\,.
\end{equation}
Here, $L_1=L\otimes \Id_n$, $L_2=\Id_n\otimes L$, the left-hand side stands for $\sum_{ijkl} \br{L_{ij},L_{kl}} E_{ij}\otimes E_{kl}$ with $E_{ij}$ the elementary matrix with only nonzero entry equal to $+1$ in position $(i,j)$, while the permutation operator is defined as 
\begin{equation}
(-)^\circ :  \gl_n(\CC)\otimes \gl_n(\CC) \to \gl_n(\CC)\otimes \gl_n(\CC)\,, \quad A\otimes B\mapsto (A\otimes B)^\circ=B\otimes A\,.
\end{equation}
The Leibniz rules for the Poisson bracket can be translated as 
\begin{equation} \label{Eq:rLder}
\begin{aligned}
  \br{A\overset{\otimes}{,}BC} =& (\Id_n\otimes B)\br{A\overset{\otimes}{,}C} + \br{A\overset{\otimes}{,}B} (\Id_n\otimes C)\,, \\
 \br{BC\overset{\otimes}{,}A} =& (B\otimes \Id_n)\br{C\overset{\otimes}{,}A} + \br{B\overset{\otimes}{,}A} (C\otimes \Id_n)\,,
\end{aligned}
\end{equation}
while antisymmetry becomes $ \br{A\overset{\otimes}{,}B}= - \br{B\overset{\otimes}{,}A}^\circ$. 
The prominent point of this formalism is that \eqref{Eq:rL} induces that the elements $(\tr L^k)$ Poisson commute due to the following chain of equalities 
\begin{equation}
 \begin{aligned}
 \frac{1}{MN}\br{\tr L^M,\tr L^N}=& (\tr \otimes \tr) (L^{M-1}\otimes L^{N-1}) \br{L\overset{\otimes}{,}L} \\
 =&(\tr \otimes \tr) \left[(L^{M-1}\otimes L^{N-1})r(L\otimes 1) -(L^{M}\otimes L^{N-1})r\right] \\
 &+(\tr \otimes \tr) \left[ (L^{M-1}\otimes L^{N-1})r^\circ (1\otimes L) - (L^{M-1}\otimes L^{N})r^\circ\right]=0\,.
 \end{aligned}
\end{equation}
Double Poisson brackets can be motivated by introducing the notation $\dgal{A,B}=\sum_{ijkl} \br{A_{ij},B_{kl}} E_{kj}\otimes E_{il}$ instead of $\br{-\overset{\otimes}{,}-}$ (note the different arrangements of indices). This operation is clearly $\CC$-linear in each argument. Antisymmetry is still written using $\dgal{-,-}$ as 
\begin{equation} \label{Eq:dbr1}
 \dgal{A,B}= - \dgal{B,A}^\circ\,,
\end{equation}
 but now the Leibniz rules become  
\begin{equation} \label{Eq:dbr2}
   \begin{aligned}
     \dgal{A,BC} =& (B\otimes \Id_n)\dgal{A,C} + \dgal{A,B} (\Id_n \otimes C)\,, \\
 \dgal{BC,A} =& (\Id_n\otimes B)\dgal{C,A} + \dgal{B,A} (C \otimes \Id_n)\,.
   \end{aligned}
\end{equation}
The Jacobi identity can also be defined using $\dgal{-,-}$, see \cite{VdB1}. Now, an analogue of \eqref{Eq:rL} is that if there exist matrices $L,(A_a)_{a\in \N}$ such that 
\begin{equation}
 \dgal{L,L}= \sum_{a\geq 0} \, (L^a \otimes A_a - A_a \otimes L^a)\,,
\end{equation}
then  the elements $(\tr L^k)$ Poisson commute due to the following chain of identities 
\begin{equation}
 \begin{aligned}
 \frac{1}{MN}\br{\tr L^M,\tr L^N}=& \sum_{ijkl} (L^{M-1})_{ji} (L^{N-1})_{lk} \dgal{L,L}_{kj,il} \\
 =&\sum_{a\geq 0}\sum_{ijkl} (L^{M-1})_{ji} (L^{N-1})_{lk} [(L^a)_{kj} (A_a)_{il} - (A_a)_{kj} (L^a)_{il}] =0\,.
 \end{aligned}
\end{equation}

For latter computations, let us mention from \cite[\S 2.4]{VdB1} that we have the following useful identities 
\begin{subequations}
 \begin{align}
  \br{\tr A, B}=& \rmm \circ \dgal{A,B}\,, \label{Eq:Tr1} \\
  \br{\tr A,\tr B}=& \tr(\rmm \circ \dgal{A,B})\,. \label{Eq:Tr2}
 \end{align}
\end{subequations}
Here, $\rmm:\gl_n(\CC)\times \gl_n(\CC)\to \gl_n(\CC)$ denotes the matrix multiplication $\rmm(A\otimes B)=AB$. We will also use the following iterated version of Leibniz rule for $A=A_1\ldots A_M$ and $B=B_1\ldots B_N$ : 
\begin{equation} \label{Eq:dbrProd}
 \dgal{A,B}= \sum_{\tau=1}^{M}\sum_{\sigma=1}^{N}
 (B_1\ldots B_{\sigma-1} \otimes A_1\ldots A_{\tau-1}) \dgal{A_\tau,B_\sigma} (A_{\tau+1}\ldots A_M \otimes B_{\sigma+1}\ldots B_N)\,.
\end{equation}

\begin{rem}
 There are major differences between the two approaches. First, double Poisson brackets are defined on non-commutative algebras, and their relation to Poisson brackets as explained above is obtained by looking at finite-dimensional representations of the algebras \cite{VdB1}. Second, we have in general that a double Poisson bracket encodes the Poisson bracket on a global phase space, while the tensor notation $\br{-\overset{\otimes}{,}-}$ is used to understand an associated Poisson bracket obtained in a suitable gauge.  
\end{rem}

\subsection{Computations on the main space} 

Fix $m\geq 2$, $I=\Z/m\Z$, $\dfat\in \N^I$ with $d_0\geq 1$, and $\ntil=(\nfat,1)$ for $\nfat\in \N^I$ with $|\nfat|>0$. 
We consider the associated quiver $Q_\dfat$ and complex Poisson manifold $\Rep(\CC \bar{Q}_\dfat, \ntil)$ as in \ref{SS:spinCMdesc}. We can express the Poisson brackets \eqref{Eq:PB} in terms of  double brackets as 
\begin{equation} \label{Eq:dbrSpace}
 \dgal{X_r,Y_s}=\delta_{rs} \, \Id_{\VV_{s+1}}\otimes \Id_{\VV_s}\,, \quad 
 \dgal{V_{r,\alpha},W_{s,\beta}}=\delta_{rs}\delta_{\alpha\beta} \,\Id_{\VV_s}\otimes \Id_{\VV_\infty}\,,
\end{equation}
and it is zero on any other pair of generators \eqref{Eq:MatGen}. Note that we see these double brackets as tensor products of square matrices of size $|\nfat|+1$, which are elements of $\End(\VV)^{\otimes 2}$. 
Introduce 
\begin{equation*}
 X:=\sum_sX_s \in \bigoplus_{s\in I} \Hom(\VV_{s+1},\VV_s)\,, \quad 
 Y:=\sum_sY_s \in \bigoplus_{s\in I} \Hom(\VV_{s-1},\VV_s)\,,
\end{equation*}
from which we note the obvious identities $X \Id_{\VV_s}=\Id_{\VV_{s-1}}X$ and $Y \Id_{\VV_s}=\Id_{\VV_{s+1}}Y$ in $\End(\VV)$. If we also introduce 
\begin{equation}
 E_r:=\sum_{s\in I} \Id_{\VV_{s+r}}\otimes \Id_{\VV_s} \in\,\End(\oplus_s\VV_s)\otimes \End(\oplus_s\VV_s)\,,
\end{equation}
we note that we can write the double brackets of the matrices $(X_s,Y_s)$ as $\dgal{X,Y}= E_1$. This implies that $\dgal{Y,X}=- E_{-1}$ by \eqref{Eq:dbr1}.

\begin{rem}
 We keep our discussion of double brackets using representations of dimension $\ntil$ of $\CC \bar{Q}_\dfat$ to simplify the exposition. In fact, all the computations that are carried out hold on $\CC \bar{Q}_\dfat$ with the double bracket of Van den Bergh  \cite[Theorem 6.3.1]{VdB1}. It can be recovered from  \eqref{Eq:dbrSpace} by replacing each matrix $(X_s,Y_s,V_{s,\alpha},W_{s,\alpha})$ by the corresponding arrow $(x_s,y_s,v_{s,\alpha},w_{s,\alpha})$, and each  $\Id_{\VV_s}$ by the idempotent $e_s$. 
\end{rem}

We are now in position to use double brackets to compute the Poisson brackets between $\Gl(\nfat)$ invariant functions on $\Rep(\CC \bar{Q}_\dfat, \ntil)$. In particular, these identities descend to the reduced space  $\Cnfat:=\Cspin$. We will repeatedly use \eqref{Eq:dbr1}, \eqref{Eq:Tr2} and \eqref{Eq:dbrProd}, while we denote $\Id_\VV$ as $1$ for simplicity.

\subsubsection{Computations for Section \ref{S:spinCM}}  

\begin{lem}  \label{Lem:dbrYY}
Denote by $\rho_{r,s}\in \{0,\ldots,m-1\}$ the representative of $s-r\in I$ for any $r,s\in I$. 
 The following identities hold for any indices : 
 \begin{subequations}
  \begin{align}
   &\br{\tr Y^{im},\tr Y^{jm}}=0\,, \quad \br{\tr Y^{im},\tr XY^{jm+1}}=-im\,\tr Y^{m(i+j)}\,; \label{Eq:YYa} \\
   &\br{\tr Y^{im},\tr W_{r,\alpha}V_{s,\beta}Y^{jm+\rho_{r,s}}}=0 \label{Eq:YYb} \,.
  \end{align}
 \end{subequations}
\end{lem}
\begin{proof}
a)  Since $\dgal{Y,Y}=0$, the first identity in \eqref{Eq:YYa} is obvious. For the second one, we note that 
 \begin{equation*}
 \begin{aligned}
  \br{\tr Y^{im},\tr XY^{jm+1}}=& (\tr\circ \rmm)\,  \sum_{\tau=0}^{im}
 (1 \otimes Y^{\tau}) \dgal{Y,X} (Y^{im-\tau-1} \otimes Y^{jm+1}) \\
=&- (\tr\circ \rmm)\,  \sum_{\tau=0}^{im}\sum_{s\in I}
 (\Id_{\VV_{s-1}} Y^{im-\tau-1} \otimes Y^{\tau} \Id_{\VV_s} Y^{jm+1}) \\
 =&- (\tr\circ \rmm)\,  \sum_{\tau=0}^{im}\sum_{s\in I}
 ( Y^{im-\tau-1} \Id_{\VV_{s+\tau}} \otimes \Id_{\VV_{s+\tau}} Y^{\tau}  Y^{jm+1}) \\
 =&-im\,\tr Y^{(i+j)m}\,.
 \end{aligned}
 \end{equation*}
 
 b) Since the double brackets of $Y$ with $Y,W_{r,\alpha}$ and $V_{s,\beta}$ are all zero, this is trivial. 
\end{proof}

\begin{lem}  \label{Lem:dbrCoord}
Let $\hat{t}^j_{r\alpha,s\beta}:=\tr W_{r,\alpha}V_{s,\beta}X^{jm+r-s}$. 
The following identities hold for any indices : 
 \begin{subequations}
  \begin{align}
   &\br{\tr X^{im},\tr X^{jm}}=0\,, \quad \br{\tr X^{im},\tr YX^{jm+1}}=im\,\tr X^{(i+j)m} \,, \label{Eq:dbrCa} \\
   &\br{\tr YX^{im+1},\tr YX^{jm+1}}=(i-j)m\,\tr YX^{(i+j)m+1}\,; \label{Eq:dbrCb} \\
   &\br{\tr X^{im},\hat{t}^j_{r\alpha,s\beta}}=0\,, \label{Eq:dbrCc} \\
   &\br{\tr YX^{im},\hat{t}^j_{r\alpha,s\beta}}=-(jm+r-s)\hat{t}^{i+j}_{r\alpha,s\beta}\,, \label{Eq:dbrCd} \\
   &\br{\hat{t}^i_{r\alpha,s\beta},\hat{t}^j_{r'\alpha',s'\beta'}}= 
   \delta_{r',s}\delta_{\beta,\alpha'}  \hat{t}^{i+j}_{r\alpha,s'\beta'} 
   -\delta_{r,s'}\delta_{\beta',\alpha}  \hat{t}^{i+j}_{r'\alpha',s\beta}\,. \label{Eq:dbrCe}
  \end{align}
 \end{subequations}
\end{lem}
\begin{proof}
 The first four identities can be computed in a way similar to the proof of Lemma \ref{Lem:dbrYY}. For \eqref{Eq:dbrCe}, we note using \eqref{Eq:dbrSpace} that 
 \begin{equation*}
  \begin{aligned}
&\br{\tr W_{r,\alpha}V_{s,\beta}X^{im+r-s}, \tr W_{r',\alpha'}V_{s',\beta'}X^{jm+r'-s'}} \\
=&+(\tr\circ \rmm)\,  (W_{r',\alpha'}\otimes 1) \dgal{W_{r,\alpha},V_{s',\beta'}} (V_{s,\beta}X^{im+r-s} \otimes X^{jm+r'-s'}) \\
&+(\tr\circ \rmm)\,  (1 \otimes W_{r,\alpha}) \dgal{V_{s,\beta},W_{r',\alpha'}} (X^{im+r-s} \otimes V_{s',\beta'} X^{jm+r'-s'}) \\
=&-\delta_{r,s'}\delta_{\alpha,\beta'} (\tr\circ \rmm)\,\, ( W_{r',\alpha'}\Id_{\VV_\infty}V_{s,\beta}X^{im+r-s} \otimes \Id_{\VV_r} X^{jm+r'-s'}) \\
&+\delta_{s,r'}\delta_{\beta,\alpha'} (\tr\circ \rmm)\,\,  (\Id_{\VV_s} X^{im+r-s} \otimes W_{r,\alpha}\Id_{\VV_\infty} V_{s',\beta'} X^{jm+r'-s'}) \\
=&-\delta_{r,s'}\delta_{\alpha,\beta'}\,\tr(W_{r',\alpha'} V_{s,\beta}X^{(i+j)m+r'-s}) \\
&+\delta_{s,r'}\delta_{\beta,\alpha'} \, \tr ( W_{r,\alpha}V_{s',\beta'} X^{(i+j)m+r-s'}) \,. 
  \end{aligned}
 \end{equation*}
This is precisely \eqref{Eq:dbrCe}. 
\end{proof}

\subsubsection{Computations for Section \ref{SS:Harm}}  

\begin{lem} 
Let $L=Y^2+\omega^2 X^2$ for some fixed $\omega \in \CC$. Then, 
  \begin{subequations}
   \begin{align}
\dgal{L,Y}=&\omega^2 [E_1(X \otimes 1) + (1 \otimes X)E_1]\,, \quad 
\dgal{L,X}=-[E_{-1}(Y \otimes 1) + (1 \otimes Y)E_{-1}]\,, \label{Eq:Lxy} \\
  \dgal{L,L}=&+\omega^2 [E_{2}(YX \otimes 1)- E_{-2} (XY\otimes 1)  + E_{2} (1\otimes XY) - E_{-2} (1\otimes YX)] \nonumber \\
  &+\omega^2 [E_3-E_{-1}](Y\otimes X) + \omega^2 [E_1-E_{-3}](X\otimes Y)\,. \label{Eq:Lz} 
   \end{align}
  \end{subequations}
If furthermore $m=1$ or $m=2$, and we set  $L_{\pm}:=Y \pm \ii \omega X$, we have   
  \begin{subequations} 
  \begin{align}
      \dgal{L,L}=&+\omega^2 E_{0} ( 1\otimes [X,Y] -  [X,Y]\otimes 1)\,, \label{Eq:LzB} \\
\dgal{L,L_+}=&-\ii \omega [E_1(L_+ \otimes 1) + (1 \otimes L_+)E_1]\,, \quad 
\dgal{L,L_-}=\ii \omega [E_1(L_- \otimes 1) + (1 \otimes L_-)E_1]\,.  \label{Eq:Lpm}
  \end{align}
  \end{subequations}
\end{lem}
\begin{proof}
  The equalities in \eqref{Eq:Lxy} are straightforward. Next, we can expand  
  \begin{equation*}
   \begin{aligned}
\dgal{L,L}=(Y\otimes 1)\dgal{L,Y} + \dgal{L,Y} (1\otimes Y) + \omega^2 (X\otimes 1)\dgal{L,X}+ \omega^2 \dgal{L,X}(1\otimes X)\,.   
   \end{aligned}
  \end{equation*}
Using \eqref{Eq:Lxy} and the identities 
\begin{equation*}
 (1\otimes X)E_r=E_{r+1}(1\otimes X),\,\,(X\otimes 1)E_r=E_{r-1}(X\otimes 1), \,\, (1\otimes Y)E_r=E_{r-1}(1\otimes Y),\,\, (Y\otimes 1)E_r=E_{r+1}(Y\otimes 1)\,, 
\end{equation*}
we get \eqref{Eq:Lz}. 

For $m=1$ or $m=2$, we can write \eqref{Eq:Lz} as \eqref{Eq:LzB}.  We can also easily derive \eqref{Eq:Lpm} from \eqref{Eq:Lxy}.
\end{proof}

We now compute some Poisson brackets between invariant functions. We restrict to the cases $m=1$ and $m=2$ as otherwise most of the functions are trivially zero. For example, since $L$ can be decomposed as linear maps $\VV_s \to \VV_{s\pm 2}$ for all $s\in I$, $\tr L$ can only be nonzero in those two cases.  

\begin{lem}  \label{Lem:dbrLL}
Denote by $\rho_{r,s}\in \{0,\ldots,m-1\}$ the representative of $r-s\in I$ for any $r,s\in I$. 
If $m=1$ or $m=2$, the following identities hold for any indices : 
 \begin{subequations}
  \begin{align}
   &\br{\tr L^{i},\tr L^{j}}=0\,; \label{Eq:LLa} \\
   &\br{\tr L,\tr W_{r,\alpha}V_{s,\beta}L_{\pm}^{jm+\rho_{r,s}}}=\mp 2\ii \omega (jm+\rho_{r,s})\tr (W_{r,\alpha}V_{s,\beta}L_{\pm}^{jm+\rho_{r,s}})  \label{Eq:LLb} \,.
  \end{align}
 \end{subequations}
\end{lem}
\begin{proof}
a) Using \eqref{Eq:LzB}, we have that 
 \begin{equation*}
 \begin{aligned}
  \br{\tr L^{i},\tr L^{j}}=& (\tr\circ \rmm)\,  \sum_{\tau=0}^{i-1}\sum_{\sigma=0}^{j-1}
 (L^\sigma \otimes L^{\tau}) \dgal{L,L} (L^{i-\tau-1} \otimes L^{j-\sigma-1}) \\
=&-\omega^2\,  \sum_{\tau=0}^{im}\sum_{\sigma=0}^{j-1}\sum_{s\in I}
\left[ \tr(L^{j-1}\Id_{\VV_s}[X,Y]L^{i-1}\Id_{\VV_s}) - \tr(L^{j-1}\Id_{\VV_s}L^{i-1}[X,Y]\Id_{\VV_s})  \right]\,.
 \end{aligned}
 \end{equation*}
 But this vanishes since $\Id_{\VV_s}L=L\Id_{\VV_{s+2}}=L\Id_{\VV_s}$ for $m=1,2$. 
 
 b) We clearly have that $\dgal{L,W_{r,\alpha}}=0$ and $\dgal{L,W_{s,\beta}}=0$. Thus by \eqref{Eq:Lpm}, if $N:=jm+\rho_{r,s}$ we have 
  \begin{equation*}
 \begin{aligned}
  \br{\tr L,\tr W_{r,\alpha}V_{s,\beta}L_{\pm}^N}=& (\tr\circ \rmm)\,  \sum_{\sigma=0}^{N-1}
 (W_{r,\alpha}V_{s,\beta}L_{\pm}^\sigma \otimes 1) \dgal{L,L_{\pm}} (1 \otimes L_{\pm}^{N-\sigma-1}) \\
=&\mp\ii \omega\,  \sum_{\tau=0}^{im}\sum_{\sigma=0}^{N-1}\sum_{s\in I}
 \tr\left(W_{r,\alpha}V_{s,\beta}L_{\pm}^\sigma (\Id_{\VV_{s+1}}L_{\pm} \Id_{\VV_{s}})L_{\pm}^{N-\sigma-1} \right) \\
 =&\mp 2\ii \omega N \,\tr (W_{r,\alpha}V_{s,\beta}L_{\pm}^{N})\,,
 \end{aligned}
 \end{equation*}
 because for $m=1,2$ we have $\Id_{\VV_{s+1}}L_{\pm}=L_{\pm} \Id_{\VV_{s}}$. 
\end{proof}

\begin{lem} \label{Lem:dbrLLm1}
 If $m=1$, we have for any $l \in \N$, 
\begin{equation}
\br{\tr L ,\tr XL^l}= -2 \,\tr YL^{l} \,, \quad
 \br{\tr L,\tr YL^l}=\,2\omega^2\, \tr XL^{l} \,. \label{Eq:LLc}
\end{equation}
 If $m=1,2$, we have for any $l \in \N$, 
 \begin{subequations}
 \begin{align}
  \br{\tr L, \tr XY L^l}=&\br{\tr L, \tr YX L^l}= -2 \tr(Y^2 L^l)+2 \omega^2 \tr (X^2 L^l)\,, \label{Eq:LLd1}\\
  \br{\tr L, \tr Y^2 L^l}=&2\omega^2 \left(\tr (XY L^l) + \tr (YX L^l) \right)\,,\label{Eq:LLd2} \\
  \br{\tr L, \tr X^2 L^l}=&-2\left(\tr (XY L^l) + \tr (YX L^l) \right)\,. \label{Eq:LLd3}
 \end{align}
 \end{subequations}
\end{lem}
\begin{proof}
We first note that for any matrix $A$, 
\begin{equation}
 \begin{aligned}
  \br{\tr L,\tr AL^l}=&(\tr\circ \rmm)\, \dgal{L,A}(1\otimes L^{l}) 
+ (\tr\circ \rmm)\, \sum_{\sigma=0}^{l-1} (AL^\sigma \otimes 1)\dgal{L,L}(1\otimes L^{l-\sigma-1})\\
=&\tr \big((\rmm \dgal{L,A}) L^{l}\big) 
+ \sum_{\sigma=0}^{l-1} \tr \big(AL^\sigma (\rmm\dgal{L,L})  L^{l-\sigma-1}\big) \\
=&\tr \big((\rmm \dgal{L,A}) L^{l}\big)\,, \label{Eq:LLtrick}
 \end{aligned}
\end{equation}
since applying the multiplication map to $\dgal{L,L}$ is zero by \eqref{Eq:LzB}. If $m=1$, we have from \eqref{Eq:Lxy} 
\begin{equation*}
 \rmm\dgal{L,Y}=2\omega^2 \sum_{s\in I} \Id_{\VV_{s+1}}X  \Id_{\VV_{s}}=2\omega^2 X\,, \quad  
  \rmm\dgal{L,X}=-2 \sum_{s\in I} \Id_{\VV_{s-1}}Y  \Id_{\VV_{s}}=-2Y\,,
\end{equation*}
so that \eqref{Eq:LLc} holds as a consequence of \eqref{Eq:LLtrick}. If $m=1,2$, we have in the same way that 
\begin{equation*}
\begin{aligned}
  &\rmm\dgal{L,XY}=\rmm\dgal{L,YX}=2\omega^2 X-2 Y^2\,, \\
  &\rmm\dgal{L,X^2}=-2(XY+YX)\,,\quad \rmm\dgal{L,Y^2}=2\omega^2 (XY+YX)\,,
\end{aligned}
\end{equation*}
and \eqref{Eq:LLd1}--\eqref{Eq:LLd3} hold as a consequence of \eqref{Eq:LLtrick}.
\end{proof}


\section{Conclusion and outlook}  \label{S:end}

In this paper, we focused on establishing  superintegrability of  complex generalisations of the rational CM system associated with cyclic quivers. These various systems are allowed to admit different types of spin variables (internal degrees of freedom) or a harmonic oscillator potential term, which are completely determined by the underlying quivers. 

To continue the investigation reported in this paper, it seems natural to try to construct such generalisations for other systems in the Calogero-Ruijsenaars family of integrable $n$-particle systems. In fact, these generalisations for the trigonometric RS system are known: for the simplest quivers considered in Section \ref{S:CM}, the corresponding systems were constructed in \cite{CF}; for the general quivers considered in Section \ref{S:spinCM}, the systems can be found in \cite{CF2,Fai,F19}. Our next aim is to unveil generalisations of the trigonometric CM and rational RS systems associated with cyclic quivers, which we expect to be maximally superintegrable. While the quivers from Section \ref{S:CM} give the usual form of these systems (see \cite[\S4.5]{GR}), the quivers from Section \ref{S:spinCM} lead to new versions of these systems endowed with different types of spin variables. In particular, we will investigate if they can be connected with the spin systems studied in \cite{CY,FP06,FP15,KLOZ}.

\vspace{1cm}


\end{document}